\documentclass{article}%
\usepackage{amsfonts}
\usepackage{amsmath}
\usepackage{amssymb}
\usepackage{graphicx}%
\newtheorem{theorem}{Theorem}

\newtheorem{corollary}[theorem]{Corollary}

\newtheorem{definition}[theorem]{Definition}
\newtheorem{example}[theorem]{Example}

\newtheorem{remark}[theorem]{Remark}

\newenvironment{proof}[1][Proof]{\noindent\textbf{#1.} }{\ \rule{0.5em}{0.5em}}
\begin{document}

\title{\textbf{On two-dimensional supersymmetric quantum mechanics, pseudoanalytic functions and transmutation operators}}
\author{Alex Bilodeau and S\'{e}bastien Tremblay \medskip\\ Department of Mathematics and Computer Science, \\ University of Quebec, Trois-Rivi\`{e}res, Qu\'{e}bec, G9A 5H7, Canada}
\maketitle

\begin{abstract}
Pseudoanalytic function theory is considered to study a two-dimensional supersymmetric quantum mechanics system. Hamiltonian components of the superhamiltonian are factorized in terms of one Vekua and one Bers derivative operators. We show that imaginary and real solutions of a Vekua equation and its Bers derivative are ground state solutions for the superhamiltonian. The two-dimensional Darboux and pseudo-Darboux transformations correspond to Bers derivatives in the complex plane. Results on the completeness of the ground states are obtained. Finally, superpotential is studied in the separable case in terms of transmutation operators. We show how Hamiltonian components of the superhamiltonian are related to the Laplacian operator using these transmutation operators.



\end{abstract}

\section{Introduction}
Supersymmetric quantum mechanics (SUSY QM) \cite{Junker, Witten} provides an interesting framework to investigate the problem of spectral equivalence of Hamiltonians, which,
historically, has been constructed as a factorization method in quantum mechanics \cite{Infeld} and as
Darboux transformations in mathematical physics \cite{Darboux, MS, Sabatier}. The SUSY QM formalism with an arbitrary ($d > 1$) dimensionality of space were constructed and investigated in and, in the last ten years, more attention has been given in the literature to the study of the two-dimensional SUSY QM \cite{Ioffe1, Ioffe2, IMGV, PSH} and the Darboux transformation for the Dirac equation \cite{Pozdeeva}. Such multidimensional models contain matrix potentials \cite{SP} which are not particularly exotic in quantum mechanics. For instance, the two-dimensional generalization of SUSY
QM considered in this paper was successfully used by Ioffe \emph{et al}. to describe the spectrum of the
Pauli operator describing spin $1/2$ fermion in the external electrostatic and magnetic field \cite{AI, IKNN, IN}.

On the other hand, theory of pseudoanalytic functions \cite{Berskniga,BersStat,Polozhy,Tutschke,Vekua} is one of the classical branches of
complex analysis extending the concepts and ideas from analytic function theory
onto a much more general situation and involving linear elliptic equations and
systems with variables coefficients. Its development in forties-fifties of the last
century was fast and deep and historically represented an important impulse to
the progress in the general theory of elliptic systems. Nevertheless the important
obstacles for a further development of pseudoanalytic function theory were its
limited practical applications together with the fact that many important results remained in the level of existence without a possibility to make them really applicable for solving problems of mathematical physics. Recent progress in
pseudoanalytic function theory reported in \cite{APFT} shows its deep relation to the
stationary Schr\"{o}dinger equation and more general linear second
order elliptic equations and includes new results which allow one to make the basic objects of
the theory and their applications fully explicit. Among other results it
is worth mentioning the possibility to obtain complete systems of solutions to second order elliptic equations with variable coefficients and their use for solving related boundary and eigenvalue problems.

In the present paper, we study known two-dimensional SUSY QM in terms of pseudoanalytic function theory. In Sections~2 and 3, we introduce preliminary notions of pseudoanalytic function theory and two-dimensional SUSY QM. We consider the two-dimensional SUSY QM studied by Ioffe \emph{et al} \cite{ABI2, ABIE, MS} which has been also considered for the Darboux transformation for the matrix Dirac equation \cite{Pozdeeva}. Then, in Section~4, Vekua equations of a special form, called main Vekua equations and related to stationary Schr\"{o}dinger equations, will be considered to study the two-dimensional SUSY QM system. Transmutation operators are introduced in Section~5 where, using these operators, we study in details separable superpotentials.

\bigskip

\section{Pseudoanalytic function theory}

This section is based on some results presented in \cite{Berskniga, BersStat, APFT}. Let $\Omega$ be a domain in $\mathbb{R}^{2}$. Throughout the whole paper we suppose that $\Omega$ is a simply connected domain and use the usual notations $z=x+\mathrm{i}y$, $\overline z=x-\mathrm{i}y$, $\partial_{z}=\displaystyle \frac{1}{2}\big(\frac{\partial}{\partial x}-\mathrm{i} \frac{\partial}{\partial y}\big)$ and $\partial_{\overline z}=\displaystyle \frac{1}{2}\big(\frac{\partial}{\partial x}+\mathrm{i} \frac{\partial}{\partial y}\big)$.

\begin{definition}
A pair of complex functions $F$ and $G$ possessing in $\Omega$ partial
derivatives with respect to the real variables $x$ and $y$ is said to be a
generating pair if it satisfies the inequality
\begin{equation}
\operatorname{Im}(\overline{F}G)\neq 0\qquad\text{in }\Omega.
\label{ImFG}
\end{equation}
\end{definition}
The inequality (\ref{ImFG}) means that $F$ and $G$ are independent in the sense that any complex-valued function $W$ defined in $\Omega$ can be uniquely expressed in the form
$$
W=\phi F+\psi G,\quad \forall z\in \Omega,
$$
where $\phi$ and $\psi$ are two real-valued functions of the variables $x$ and $y$. In other words, the pair $(F,G)$ generalizes the pair $(1,\mathrm{i})$ corresponding to the usual complex analytic function
theory. Sometimes it is convenient to associate with the function $W$
the function $\omega=\phi+\mathrm{i}\psi$. The correspondence between $W$ and $\omega$
is one-to-one.

\begin{definition}
Let $(F,G)$ be a generating pair in $\Omega$ and the function $W=\phi F+\psi G$ be defined in a neighborhood of $z_0\in \Omega$. We say that the function $W$ possesses the $(F,G)$-derivative $\overset{\circ}{W}$ at $z_0$ if the (finite) limit
$$
\overset{\circ}{W}(z_0)=\frac{\mathrm{d}_{(F,G)}W}{\mathrm{d}z}\Big|_{z=z_0}=\lim_{z\rightarrow z_0}\frac{W(z)-\phi(z_0)F(z)-\psi(z_0)G(z)}{z-z_0}
$$
exists, where $\phi(z_0)$ and $\psi(z_0)$ are the unique real constants such that $W(z_0)=\phi(z_0) F(z_0)+\psi(z_0) G(z_0)$.
\end{definition}

\begin{theorem}\cite{Berskniga}
\label{FGderivativetheo}
Let $(F,G)$ be a generating pair in some
open domain $\Omega $ and $W\in C^1(\Omega)$. The $(F,G)$-derivative
$\overset{\circ}{W}$ exists and has the form
\begin{equation}
\label{generalFGderiv}
\overset{\circ}{W}=(\partial_z\phi) F+(\partial_z\psi) G=\partial_z W-A W-B \overline{W}
\end{equation}%
if and only if
\begin{equation}
\partial_{\overline{z}}W=a W+b \overline{W},
\label{vekua}
\end{equation}
where $a,\ b,\ A$ and $B$ are called the \textit{characteristic coefficients} associated with the pair $(F,G)$ in $\Omega$, defined by the formulas
\begin{align*}
a=a_{(F,G)} =&-\displaystyle \frac{\overline{F}\partial_{\overline{z}}G-
\overline{G} \partial_{\overline{z}}F}{F\overline{G}-\overline{F}G}, & b=b_{(F,G)} =&\displaystyle%
\frac{F\partial_{\overline{z}}G-G\partial_{\overline{z}}F}{F\overline{G}-\overline{F}G}
\\*[2ex]
A=A_{(F,G)}=&-\displaystyle\frac{\overline{F}\partial_z G-\overline{G}\partial_z F}{F%
\overline{G}-\overline{F}G}, & B=B_{(F,G)}=&\displaystyle\frac{F\partial_zG-G\partial_z F}{F%
\overline{G}-\overline{F}G}.%
\end{align*}
\end{theorem}
Notice that $F\overline{G}-\overline{F}G=-2\mbox{i} \operatorname{Im}(\overline{F}G)\neq 0$ from (\ref{ImFG}).

Equation (\ref{vekua}), generalizing the Cauchy-Riemann system, is called the \textbf{Vekua equation}. A function such that $\overset{\circ}{W}$ exists everywhere on $\Omega$ is
called $(F,G)$-\textbf{pseudoanalytic function}.

\begin{remark}
We notice that the Vekua equation~(\ref{vekua}) can be rewritten in the form $\phi_{\bar{z}}F+\psi_{\bar{z}}G=0$. Also, the functions $F$ and $G$ are $(F,G)$-pseudoanalytic functions where $\overset{\circ}{F}\equiv 0\equiv \overset{\circ}{G}$.
\end{remark}

\begin{definition}
\label{DefSuccessor}
Let $(F,G)$ and $(F_{1},G_{1})$ - be two generating pairs
in $\Omega$. $(F_{1},G_{1})$ is called successor of $(F,G)$ and $(F,G)$ is
called predecessor of $(F_{1},G_{1})$ if%
\begin{equation}
a_{(F_{1},G_{1})}=a_{(F,G)}\qquad\text{and}\qquad b_{(F_{1},G_{1})}%
=-B_{(F,G)}.
\label{coeffderiv}
\end{equation}
\end{definition}
This definition arises naturally in relation to the notion of the $(F,G)$-derivative due to
the following fact.

\begin{theorem}
\label{ThBersDer}Let $W$ be an $(F,G)$-pseudoanalytic function and let
$(F_{1},G_{1})$ be a successor of $(F,G)$. Then $\overset{\circ}{W}$ is an
$(F_{1},G_{1})$-pseudoanalytic function, i.e.
$$
\partial_{\overline{z}}\overset{\circ}{W}=a_{(F_1,G_1)}\overset{\circ}{W}+b_{(F_1,G_1)}\overline{\overset{\circ}{W}}=
a_{(F,G)}\overset{\circ}{W}-B_{(F,G)}\overline{\overset{\circ}{W}}.
$$
\end{theorem}
This process of construction of new Vekua
equations associated with the previous ones via relations (\ref{coeffderiv}) can be continued and we arrive
at the following definition.

\begin{definition}
\label{DefSeq}
A sequence of generating pairs $\left\{  (F_{m},G_{m})\right\}$, for $m\in \mathbb{Z}$, is called a generating sequence if $(F_{m+1}%
,G_{m+1})$ is a successor of $(F_{m},G_{m})$. If $(F_{0},G_{0})=(F,G)$, we say
that $(F,G)$ is embedded in $\left\{  (F_{m},G_{m})\right\}  $.
\end{definition}
We say that two generating pairs, $(F, G)$ and $(\widetilde F, \widetilde G)$ are called equivalent if $\widetilde F =a_{11}F + a_{12}G$ and $\widetilde G = a_{21}F + a_{22}G$, the $a_{jk}$ being real constants. A generating sequence $\left\{  (F_{m},G_{m})\right\}$ is said to have period $\mu>0$ if $(F_{m+\mu},G_{m+\mu})$ is equivalent to $(F_{m},G_{m})$, that is their characteristic coefficients coincide.

Let $W$ be an $(F,G)$-pseudoanalytic function. Using a generating sequence in
which $(F,G)$ is embedded we can define the higher derivatives of $W$ by the
recursion formula%
\[
W^{[0]}=W;\qquad W^{[m+1]}=\frac{\mathrm{d}_{(F_{m},G_{m})}W^{[m]}}{\mathrm{d}z},\quad
m\in \mathbb{Z}_{\geq 0}.
\]
Let $(F,G)$ be a generating pair in $\Omega$. Its adjoint generating
pair $(F,G)^{\ast}=(F^{\ast},G^{\ast})$ is defined by the following formulas
\[
F^{\ast}=-\frac{2\overline{F}}{F\overline{G}-\overline{F}G},\qquad G^{\ast
}=\frac{2\overline{G}}{F\overline{G}-\overline{F}G}.
\]
The $(F,G)$-integral is then given by
\begin{equation}
\label{FGintegration}
\int_{\Gamma}W\mathrm{d}_{(F,G)}z=F(z_{1})\operatorname{Re}%
\int_{\Gamma}G^{\ast}W\mathrm{d}z+G(z_{1})\operatorname{Re}\int_{\Gamma}F^{\ast
}W\mathrm{d}z,
\end{equation}
where $\Gamma$ is a rectifiable curve leading from $z_{0}$ to $z_{1}$.

If $W=\phi F+\psi G$ is an $(F,G)$-pseudoanalytic function where $\phi$ and
$\psi$ are real valued functions then
\[
\int_{z_{0}}^{z}\overset{\circ}{W}\mathrm{d}_{(F,G)}z=W(z)-\phi(z_{0})F(z)-\psi
(z_{0})G(z). \label{FGAnt}%
\]
This integral is
path-independent and represents the $(F,G)$-antiderivative of $\overset{\circ}{W}$.

Definition~\ref{DefSeq} of generating sequence is the main ingredient for obtaining the explicit form of formal powers for a certain Vekua equation. Briefly speaking, formal powers are solutions of a Vekua equation (\ref{vekua}) generalizing the usual analytic powers $\big\{(z-z_0)^n\big\}_{n=0}^\infty$ in the sense that locally when $z\rightarrow z_0$ they behave asymptotically like the usual powers and under some additional conditions on the coefficients $a_{(F,G)}$ and $b_{(F,G)}$ they form a complete system
in the space of all solutions of the Vekua equation in the same sense as the analytic powers $\big\{(z-z_0)^n\big\}_{n=0}^\infty$ form a complete system in the space of analytic functions.

\begin{definition}
\label{DefFormalPower}The formal power $Z_{m}^{(0)}(a,z_{0};z)$ with center at
$z_{0}\in\Omega$, coefficient $a$ and exponent $0$ is defined as the linear
combination of the generators $F_{m}$, $G_{m}$ with real constant coefficients
$\lambda$, $\mu$ chosen so that $\lambda F_{m}(z_{0})+\mu G_{m}(z_{0})=a$. The
formal powers with exponents $n\in \mathbb{Z}_{>0} $ are defined by the recursion
formula%
\begin{equation}
Z_{m}^{(n+1)}(a,z_{0};z)=(n+1)\int_{z_{0}}^{z}Z_{m+1}^{(n)}(a,z_{0}%
;\zeta)\mathrm{d}_{(F_{m},G_{m})}\zeta. \label{recformula}%
\end{equation}
\end{definition}

The following properties can be derived from this last definition.

\begin{enumerate}
\item $Z_{m}^{(n)}(a,z_{0};z)$ is an $(F_{m},G_{m})$-pseudoanalytic function
of $z$.

\item If $a_1$ and $a_2$ are two real constants, then
$Z_{m}^{(n)}(a_1+\mathrm{i}a_2,z_{0};z)=a_1Z_{m}%
^{(n)}(1,z_{0};z)+a_2Z_{m}^{(n)}(\mathrm{i},z_{0};z).$

\item The formal powers satisfy the differential relations%
\[
\frac{\mathrm{d}_{(F_{m},G_{m})}Z_{m}^{(n)}(a,z_{0};z)}{\mathrm{d}z}=nZ_{m+1}^{(n-1)}%
(a,z_{0};z).
\]

\item The asymptotic formulas
\[
\lim_{z\rightarrow z_{0}} Z_{m}^{(n)}(a,z_{0};z)= a(z-z_{0})^{n},
\]
hold.
\end{enumerate}
It follows from (\ref{recformula}) that once the generating sequence  $\left\{  (F_{m},G_{m})\right\}$ is known, the formal powers $Z_{m}^{(n)}(a,z_{0};z)$, $n\in \mathbb{Z}_{> 0}$, can be obtained by successive integrations. The scheme of the $(F_m,G_m)$-derivatives is the following:
\[
\begin{array}{cccccccc}
\vdots && \vdots && \vdots&& \vdots & \\
Z^{(2)} && Z_1^{(2)} && Z_2^{(2)}&& Z_3^{(2)} & \cdots\\
&\searrow^{\frac{\mathrm{d}_{(F,G)}}{\mathrm{d}z}}&&\searrow^{\frac{\mathrm{d}_{(F_1,G_1)}}{\mathrm{d}z}}&&\searrow^{\frac{\mathrm{d}_{(F_2,G_2)}}{\mathrm{d}z}}& & \\
 Z^{(1)} && Z_1^{(1)} && Z_2^{(1)}&& Z_3^{(1)} & \cdots\\
  &\searrow^{\frac{\mathrm{d}_{(F,G)}}{\mathrm{d}z}}&&\searrow^{\frac{\mathrm{d}_{(F_1,G_1)}}{\mathrm{d}z}}&&\searrow^{\frac{\mathrm{d}_{(F_2,G_2)}}{\mathrm{d}z}}& & \\
Z^{(0)} && Z_1^{(0)} && Z_2^{(0)} && Z_3^{(0)}&  \cdots
\end{array}
\]

\begin{definition}
Let $W(z)$ be a given $(F,G)$-pseudoanalytic function defined for small values
of $\left\vert z-z_{0}\right\vert $. The series
\[
\sum_{n=0}^{\infty}Z^{(n)}(a_n,z_{0};z) \label{Taylorseries}%
\]
with the coefficients given by
$$
a_{n}=\frac{W^{[n]}(z_{0})}{n!}
$$
is called the Taylor series
of $W(z)$ at $z_{0}$, formed with formal powers.
\end{definition}

\begin{definition}
A generating pair $(F,G)$ is called complete if these functions are defined and satisfy the H\"{o}lder condition for all finite values of $z$, the limits $F(\infty)$ and $G(\infty)$ exist, $\operatorname{Im\big(\overline{F(\infty)}G(\infty)\big)>0}$ and the functions $F(1/z)$, $G(1/z)$ also satisfy the H\"{o}lder condition. A complete generating pair is called normalized if $F(\infty)=1$ and $G(\infty)=\mathrm i$.
\end{definition}

For now on we assume that $(F,G)$ is a complete normalized generating pair. Then the following completeness results were obtained. (Following \cite{Berskniga}, we shall say that a sequence of functions $f_n$ converges normally in a domain $\Omega$ if it converges uniformly on every bounded closed subdomain of $\Omega$.)

\begin{theorem}[Expansion theorem \cite{BersApprox}]
\label{Expansion}
Let $W$ be an $(F,G)$-pseudoanalytic function defined for $|z-z_0|< R$. Then it admits a unique expansion of the form $W(z)=\sum_{n=0}^\infty Z^{(n)}(a_n,z_0;z)$ which converges normally for $|z-z_0|< \theta R$, where $\theta$ is a positive constant depending on the generating sequence.
\end{theorem}

\begin{theorem}[Runge's approximation \cite{BersApprox}]
\label{Runge}
A pseudoanalytic function defined in a simply connected domain can be approached by a normally convergent sequence of formal polynomials (linear combinations of formal powers with positive exponents).
\end{theorem}
\section{Two-dimensional SUSY QM}
In \cite{Witten} Witten proposed the conventional one-dimensional SUSY QM characterized
by the simplest realization of the supersymmetric algebra

\begin{equation}
\label{susyAlgebra}
\{\hat Q^+,\hat Q^-\}=\hat H,\quad [\hat H,\hat Q^{\pm}]=0,\quad \{\hat Q^{+},\hat Q^{+}\}=\{\hat Q^{-},\hat Q^{-}\}=0,
\end{equation}
where $\{A,B\}=AB+BA$ and $[A,B]=AB-BA$. The operator $\widehat H$ is the superhamiltonian represented by the diagonal matrix $\mathrm{diag}(H^{(0)},H^{(1)})$, with $H^{(0)}=q^+q^-=-\partial^2+V^{(0)}$, $H^{(1)}=q^-q^+=-\partial^2+V^{(1)}$, $q^{\pm}=\mp \partial+\partial \chi$ and $\partial\equiv \frac{\mathrm{d}}{\mathrm{d}x}$. The supercharge operators $\hat Q^{\pm}$ are off diagonal matrices with elements $q^{\pm}$. The function $\chi$ called the superpotential is a smooth function defined via the zero-energy wave function $\psi_0\equiv \exp(-\chi)$ of the Schr\"{o}dinger equation $H^{(0)}$, i.e. $H^{(0)}\psi_0=0$ (possibly after an appropriate choice of the origin for the energy scale). The (anti)commutation relations of SUSY algebra (\ref{susyAlgebra}) mean, respectively, factorization of Hamiltonians, intertwining of $H^{(i)}$ with $q^{\pm}$ and nilpotent structure of supercharges.

Let us now consider the two-dimensional case where for convenience we use the notations $(x,y)\equiv (x_1,x_2)$, $\partial/\partial x\equiv \partial_1$ and $\partial/\partial y\equiv \partial_2$. The direct
generalization of SUSY QM \cite{ABI2, ABIE, MS} satisfies the conventional Witten's SUSY algebra (\ref{susyAlgebra}). The $4\times 4$ superhamiltonian is realized by the block diagonal matrix
$$
\hat H=\left(
\begin{array}{ccc}
H^{(0)} & 0 & 0  \\
0 & H_{ij}^{(1)} & 0 \\
0 & 0 & H^{(2)}
\end{array}
\right),
$$
with the supercharges operators
$$
\hat Q^+=\left(
\begin{array}{cccc}
0 & 0 & 0 & 0 \\
q_1^- & 0 & 0 & 0 \\
q_2^- & 0 & 0 & 0 \\
0 & p_1^+ & p_2^+ & 0
\end{array}
\right)\quad \mbox{and}\quad \hat Q^-=(\hat Q^+)^\dagger=\left(
\begin{array}{cccc}
0 & q_1^+ & q_2^+ & 0 \\
0 & 0 & 0 & p_1^- \\
0 & 0 & 0 & p_2^- \\
0 & 0 & 0 & 0
\end{array}
\right),
$$
where
\begin{equation}
\label{operatorsqp}
q_i^{\pm}=\mp \partial_i + \partial_i \chi, \quad\quad  p_i^{\pm}=\displaystyle \sum_{k=1}^2 \epsilon_{ik}q_k^{\mp},\quad\quad i,j\in\{1,2\}.
\end{equation}
Here $\chi\in C^2(\Omega,\mathbb R)$ is a real-valued function of the variables $\mathbf x=(x_1,x_2)$ and $\epsilon_{ik}$ is the fundamental antisymmetrical tensor. In particular, we observe that $(\hat Q^{\pm})^2=0$ from (\ref{susyAlgebra}) imply that
\begin{equation}
\label{relationspq}
\displaystyle \sum_{k=1}^2 p_k^+q_k^-=0\qquad \mbox{ and }\qquad \displaystyle \sum_{k=1}^2 q_k^+p_k^-=0.
\end{equation}
Moreover, a direct calculation from (\ref{operatorsqp}) shows that
$$
[q_i^-,q_j^+]=2\partial_i\partial_j \chi.
$$
From the first equality $\{\hat Q^+,\hat Q^-\}=\hat H$ in (\ref{susyAlgebra}) the two scalar Schr\"{o}dinger operators $H^{(0)}$, $H^{(2)}$ and the $2 \times 2$ matrix Schr\"{o}dinger operator
$H_{ij}^{(1)}$ are expressed in terms of the components of supercharges:
\begin{equation}
\label{superHamiltonians}
\begin{array}{rcl}
H^{(0)}&=&\displaystyle \sum_{k=1}^2 q_k^+q_k^-=-\triangle + U^{(0)}=-\triangle+\displaystyle \sum_{k=1}^2 (\partial_k \chi)^2-\partial^2_k \chi, \\*[2ex]
H^{(2)}&=&\displaystyle \sum_{k=1}^2 p_k^+p_k^-=-\triangle + U^{(2)}=-\triangle+\displaystyle \sum_{k=1}^2 (\partial_k \chi)^2+\partial^2_k \chi, \\*[2ex]
H^{(1)}_{ij}&=& q_i^-q_j^++p_i^-p_j^+=\delta_{i,j}H^{(0)}+2\partial_i\partial_j \chi,
\end{array}
\end{equation}
where $\delta_{i,j}$ is the Kronecker delta, $\triangle\equiv \partial_1^2+\partial_2^2$ and the zero-energy
wave functions of the scalar Hamiltonians $H^{(0)}$, $H^{(2)}$ are written as $\psi^{(0)}_0=\mathrm e^{-\chi}$ and $\psi^{(2)}_0=\mathrm e^\chi$, respectively.

The quasi-factorization in (\ref{superHamiltonians}) ensures that the commutation relations in (\ref{susyAlgebra}) hold, and leads to the following intertwining relations for the components of the superhamiltonian $\widehat H$:

\begin{equation}
\label{intertwining}
\begin{array}{c}
H^{(0)}q_i^+=\displaystyle \sum_{k=1}^2 q_k^+H^{(1)}_{ki},\quad\quad q_i^-H^{(0)}=\displaystyle \sum_{k=1}^2 H^{(1)}_{ik}q_k^-,\\*[2ex]
H^{(2)}p_i^+=\displaystyle \sum_{k=1}^2 p_k^+H^{(1)}_{ki},\quad\quad p_i^-H^{(2)}=\displaystyle \sum_{k=1}^2 H^{(1)}_{ik}p_k^-.
\end{array}
\end{equation}
These intertwining
relations (\ref{intertwining}) imply that components of the vector wave functions of the matrix Hamiltonian $H^{(1)}_{ij}$ are connected (up to a constant) with the wave functions of the scalar Hamiltonians $H^{(0)}$ and $H^{(2)}$:
\begin{equation}
\label{equivalence}
\begin{array}{c}
\psi^{(1)}_i(\mathbf x;E)=q_i^- \psi^{(0)}(\mathbf x;E), \quad \psi^{(1)}_i(\mathbf x;E)=p_i^- \psi^{(2)}(\mathbf x;E),\\*[2ex]
\psi^{(0)}(\mathbf x;E)=\displaystyle \sum_{k=1}^2 q_k^+ \psi^{(1)}_k(\mathbf x;E),\quad \psi^{(2)}(\mathbf x;E)=\displaystyle \sum_{k=1}^2 p_k^+ \psi^{(1)}_k(\mathbf x;E).
\end{array}
\end{equation}
We remark that $H^{(0)}$ and $H^{(2)}$ are not intertwined with each other and are not (in general) isospectral, nevertheless relations (\ref{equivalence}) lead
to the connections between the spectrum of the matrix Hamiltonian $H^{(1)}$ and the spectra of the two
scalar ones $H^{(0)},H^{(2)}$
\begin{equation}
\label{H}
H\Psi=E\Psi,\quad H=\mathrm{diag}(H^{(0)},H^{(2)}),\quad \Psi=(\psi^{(0)},\psi^{(2)})^\top
\end{equation}
and
\begin{equation}
\label{Htilde}
H^{(1)}\Phi=E \Phi,\quad \Phi=(\psi_1^{(1)}, \psi_2^{(1)})^\top.
\end{equation}
Here `equivalence' means coincidence of spectra up to zero modes of the operators  $q_i^{\pm}$, $p_i^{\pm}$.
Using relations (\ref{equivalence}) the two-dimensional Darboux transformations \cite{MS, Sabatier} $D$ and $D^{\dagger}$  are defined by
\begin{equation}
\label{2DDarbouxTransf}
\Phi=D \Psi,\quad\quad D=\left(
\begin{array}{cc}
q_1^- & p_1^- \\
q_2^- & p_2^-
\end{array}
\right)
\end{equation}
and
\begin{equation}
\label{2DDarbouxTransf2}
\Psi=D^{\dagger} \Phi,\quad\quad D^{\dagger}=\left(
\begin{array}{cc}
q_1^+ & q_2^+ \\
p_1^+ & p_2^+
\end{array}
\right),
\end{equation}
where
\begin{equation}
\label{DDH}
D^{\dagger}D=H, \qquad DD^{\dagger}=H^{(1)}.
\end{equation}
Moreover, pseudo-Darboux transformations $D_1$ and $D_1^{\dagger}$ can be defined as
\begin{equation}
\label{D1}
D_1=\left(
\begin{array}{cc}
p_2^- & p_1^- \\
q_2^- & q_1^-
\end{array}
\right)\qquad \mbox{ and } \qquad D_1^{\dagger}=\left(
\begin{array}{cc}
p_2^+ & q_2^+ \\
p_1^+ & q_1^+
\end{array}
\right),
\end{equation}
where $D_1D_1^{\dagger}=H$, $D_1^{\dagger} D_1=\widetilde H^{(1)}$ and $\widetilde H^{(1)}:=\left(
\begin{array}{cc}
H_{11}^{(1)} & -H_{12}^{(1)} \\
-H_{21}^{(1)} & H_{22}^{(1)}
\end{array}
\right)$.

\section{Pseudoanalytic functions and SUSY QM}
Let us define first the following complex operators:
\begin{equation}
\label{operatorsV}
\begin{array}{rr}
V=\partial_{\overline z}-\partial_{\overline z}\chi\,C, & \overline V=\partial_{z}-\partial_{z}\chi\,C, \\*[2ex]
\quad V_1=\partial_{\overline z}+\partial_{z}\chi\,C, & \overline{V}_1=\partial_{z}+\partial_{\overline z}\chi\,C,
\end{array}
\end{equation}
together with the operators projecting the real and imaginary parts
$$
P^+=\frac{1}{2}(I+C),\qquad P^-=\frac{1}{2\mathrm i}(I-C),
$$
where $C$ represents the complex conjugate operator. Operators (\ref{operatorsV}) are related by the complex conjugate operator $C$ as
\begin{equation}
\label{CVC}
CV=\overline VC \qquad \text{ and }\qquad CV_1=\overline V_1C.
\end{equation}
Defining now the operator $P:\mathbb {C}\rightarrow \mathbb {R}^2$ by $Pw=(P^-, P^+)^\top w$, we obtain the following result.

\begin{theorem}
The following operator equalities hold:
\begin{equation}
\label{DP=2PV}
DP=2P\overline V, \qquad D^{\dagger}P=-2PV_1,\qquad D_1P=2P\overline V_1,\qquad D_1^{\dagger}P=-2PV,
\end{equation}
for any differentiable complex-valued function in $\Omega$. Moreover, the following factorization are obtained
\begin{equation}
\label{HP=4PVV}
\begin{array}{rclll}
HP&=&-4PV_1\overline V& \qquad H^{(1)}P=-4P\overline VV_1,& \widetilde H^{(1)}P=-4PV\overline V_1,\\
&=& -4P\overline V_1 V,
\end{array}
\end{equation}
for any twice differentiable complex-valued function in $\Omega$.
\end{theorem}
\begin{proof}
Let $w=w_1+\mathrm i w_2$ be differentiable in $\Omega$, then
$$
\begin{array}{rcl}
DPw&=& \left(
         \begin{array}{cc}
           q_1^- & q_2^+ \\
           q_2^- & -q_1^+ \\
         \end{array}
       \right)\left(
                \begin{array}{c}
                  w_2 \\
                  w_1 \\
                \end{array}
              \right)=\left(
                        \begin{array}{c}
                          (\partial_1+\partial_1\chi)w_2+(-\partial_2+\partial_2\chi)w_1 \\
                          (\partial_2+\partial_2\chi)w_2-(-\partial_1+\partial_1\chi)w_1 \\
                        \end{array}
                      \right)\\*[2ex]

&=&P\Big[(\partial_1-\mathrm i\partial_2)(w_1+\mathrm iw_2)-(\partial_1-\mathrm i\partial_2)\chi \cdot (w_1-\mathrm iw_2)\Big]\\*[2ex]
&=& 2P\overline Vw.
\end{array}
$$
Other equalities in (\ref{DP=2PV}) are shown in a similar way and operator equalities (\ref{HP=4PVV}) are the results of (\ref{DDH}) and (\ref{DP=2PV}).
\end{proof}
\begin{remark}
We note in particular that for any \emph{real-valued} function in $C^2(\Omega)$ we obtain the factorizations
$$
\begin{array}{rcr}
H^{(0)}=4\mathrm i V_1\overline V\mathrm i=4\mathrm i \overline V_1 V\mathrm i, && H^{(2)}=-4V_1\overline V=-4\overline V_1V, \\*[2ex]
H^{(1)}_{11}+\mathrm i H^{(1)}_{12}=4\mathrm i V \overline V_1\mathrm i, && H^{(1)}_{22}+\mathrm i H^{(1)}_{21}=-4\overline V V_1,\\*[2ex]
\widetilde H^{(1)}_{11}+\mathrm i \widetilde H^{(1)}_{12}=4\mathrm i\overline V V_1\mathrm i, && \widetilde H^{(1)}_{22}+\mathrm i \widetilde H^{(1)}_{21}=-4 V \overline V_1.
\end{array}
$$
\end{remark}

\begin{remark}
Operators $H$ and $H^{(1)}$ can also be written in terms of complex operators $\partial_z,\partial_{\overline z}$ as
\begin{equation}
\label{Hcomplex}
H=-\triangle +4\left(
\begin{array}{cc}
|\partial_z \chi|^2-\partial_{\overline z}\partial_z \chi &  \\
 & |\partial_z \chi|^2+\partial_{\overline z}\partial_z \chi
\end{array}
\right)
\end{equation}
and
\begin{equation}
\label{Htildecomplex}
H^{(1)}=-\triangle+4\left(
\begin{array}{cc}
|\partial_z \chi|^2+\operatorname{Re}\partial^2_z \chi & -\operatorname{Im}\partial_z^2 \chi  \\
 -\operatorname{Im}\partial_z^2 \chi & |\partial_z \chi|^2-\operatorname{Re}\partial^2_z \chi
\end{array}
\right).
\end{equation}
\end{remark}

\begin{remark}
From equation (\ref{Hcomplex}) we see that potentials $U^{(0,2)}$ defined in (\ref{superHamiltonians}) can be written in the form
$$
\partial_{\overline z}R+|R|^2=\frac{1}{4}U^{(0,2)},
$$
where $R\equiv -\partial_z \chi$ for the potential $U^{(0)}$ and  $R\equiv \partial_z \chi$ for the potential $U^{(2)}$. This equation corresponds to the complex Riccati equation (see \cite{APFT} for more details on this equation).
\end{remark}

Let us now consider the Vekua equation $VW=0$, i.e.
\begin{equation}
\label{mainvekua}
\partial_{\overline z}W=\partial_{\overline z}\chi\,\overline{W}\quad \text{ in }\Omega.
\end{equation}
We call this equation the \textbf{main Vekua equation}. A generating pair for equation (\ref{mainvekua}) is given by
\begin{equation}
\label{(F,G)}
(F,G)=(\mathrm e^{\chi},\ \mathrm{i}\mathrm e^{-\chi}),
\end{equation}
such that $a_{(F,G)}=A_{(F,G)}=0$, $b_{(F,G)}=\partial_{\overline z}\chi$ and $B_{(F,G)}=\partial_{z}\chi$. In this case, we note that from (\ref{generalFGderiv}) the $(F,G)$-derivative is given by
$$
\overset{\circ}{W}=\frac{\mathrm d_{(F,G)}W}{\mathrm dz}=(\partial_z -\partial_{z}\chi\,C)W=\overline VW=(CVC)W,
$$
where second and last equalities come, respectively, from  (\ref{generalFGderiv}) and (\ref{CVC}). Also from (\ref{vekua}) and (\ref{coeffderiv}) we observe that for any given successor $(F_1,G_1)$ of the pair $(F,G)$ the $(F_1,G_1)$-pseudoanalytic functions satisfy Vekua equation
$$
V_1w=0,
$$
where, in particular from theorem \ref{ThBersDer}, we have $V_1\overset{\circ}{W}=0$ when $W$ satisfies the main Vekua equation (\ref{mainvekua}).

\begin{theorem}
\label{solutionVekuaSusy}
Let $W\in C^3(\Omega)$ be a solution of the main Vekua equation (\ref{mainvekua}). Then $\Psi=PW$ and $\Phi=P\overset{\circ}{W}$ are wave function solutions in the kernel of $H$ and $H^{(1)}$, respectively.
\end{theorem}

\begin{proof}
Since $W$ satisfies the main Vekua equation $VW=0$ we know that $\overset{\circ}{W}$ exists and is given by $\overset{\circ}{W}=\overline VW$. Moreover, we have that $V_1\overset{\circ}{W}=0$ such that
$$
0=V_1\overline V W \quad \Rightarrow \quad 0=PV_1\overline V W=-\frac{1}{4}HPW,
$$
where the last equality comes from (\ref{HP=4PVV}). Also, using again (\ref{HP=4PVV}) we easily find that
$$
H^{(1)}P\overset{\circ}{W}=-4P\overline V V_1\overset{\circ}{W}=0.
$$
\end{proof}

This theorem can be summarized in the following commutative diagram:
\begin{equation}
\label{diagram}
\begin{array}{rcc}
W\in \ker V & \overset{P}{\longrightarrow} & \Psi=(\operatorname{Im}W,\ \operatorname{Re}W)^\top \in \ker H \\
\\
\overline V=\displaystyle \frac{\mathrm d_{(F,G)}}{\mathrm dz}\downarrow && \downarrow D \\
\\
\overset{\circ}{W}\in \ker V_1 & \overset{P}{\longrightarrow} & \Phi=(\operatorname{Im}\overset{\circ}{W},\ \operatorname{Re}\overset{\circ}{W})^\top \in \ker H^{(1)}.
\end{array}
\end{equation}

Note that
the operator $\partial_{\overline{z}}$ applied to a real-valued function
$\phi$ can be regarded as a kind of gradient, and if we know that
$\partial_{\overline{z}}\phi=\Phi$ in a whole complex plane or in a convex
domain, where $\Phi=\Phi_{1}+\mathrm i\Phi_{2}$ is a given complex valued function
such that its real part $\Phi_{1}$ and imaginary part $\Phi_{2}$ satisfy the
equation

\begin{equation}
\partial_{2}\Phi_{1}-\partial_{1}\Phi_{2}=0, \label{casirot}%
\end{equation}
then we can reconstruct $\phi$ up to an arbitrary real constant $c$ in the
following way%
\begin{equation}
\phi(x,y)=2\left(  \int_{a}^{x}\Phi_{1}(\xi,b)\mathrm d\xi+\int_{b}^{y}%
\Phi_{2}(a,\eta)\mathrm d\eta\right)  +c \label{Antigr}%
\end{equation}
where $(a,b)$ is an arbitrary fixed point in the domain of interest.
Note that this formula can be easily extended to any simply connected domain
by considering the integral along an arbitrary rectifiable curve $\Gamma$
leading from $(a,b)$ to $(x,y)$%
\[
\phi(x,y)=2\left(  \int_{\Gamma}\Phi_{1}\mathrm dx+\Phi_{2}\mathrm dy\right)  +c.
\]
By $\overline{A}$ we denote this integral operator:%
\[
\overline{A}[\Phi](x,y)=2\left(  \int_{\Gamma}\Phi_{1}\mathrm dx+\Phi_{2}\mathrm dy\right).
\]
Thus if $\Phi$ satisfies (\ref{casirot}), there exists a family of real valued
functions $\phi$ such that $\partial_{\overline{z}}\phi=\Phi$, given by the
formula $\phi=\overline{A}[\Phi]$.

In a similar way we introduce the integral operator
$$
A[\Phi](x,y)=2\left(  \int_{\Gamma}\Phi_{1}\mathrm dx-\Phi_{2}\mathrm dy\right),
$$
corresponding to the operator $\partial_z$ and applied to complex functions whose real and imaginary parts satisfy the condition
$$
\partial_{2}\Phi_{1}+\partial_{1}\Phi_{2}=0.
$$

\begin{theorem}
\label{PrTransform}\cite{Krpseudoan} Let $W_{1}\in \ker H^{(2)}$ in a simply connected domain $\Omega$. Then the real valued
function $W_{2}\in \ker H^{(0)}$ such that $W=W_{1}+\mathrm i W_{2}$ is a
solution of (\ref{mainvekua}), is constructed according to the formula%
\begin{equation}
W_{2}=\mathrm e^{-\chi}\overline{A}[\mathrm i\mathrm e^{2\chi}\partial_{\overline{z}}(\mathrm e^{-\chi}W_{1})].
\label{transfDarboux}%
\end{equation}

Given a function $W_{2}\in \ker H^{(0)}$, the corresponding function $W_{1}\in \ker H^{(2)}$ such that $W=W_{1}+\mathrm i W_{2}$ is a solution of
(\ref{mainvekua}), is constructed as follows%
\begin{equation}
W_{1}=-\mathrm e^{\chi}\overline{A}[\mathrm i\mathrm e^{-2\chi}\partial_{\overline{z}}(\mathrm e^{\chi}W_{2})].
\label{transfDarbouxinv}%
\end{equation}
\end{theorem}
\begin{remark}
Note that when $\chi\equiv 0$ equations (\ref{transfDarboux}) and (\ref{transfDarbouxinv}) are the well-known formulas in complex analysis for constructing conjugate harmonic functions.
\end{remark}


 The following statements are direct corollaries of theorems \ref{solutionVekuaSusy}, \ref{PrTransform} and of the convergence theorems \ref{Expansion} and \ref{Runge} from Bers \cite{Berskniga, BersApprox}. We suppose that $\Omega$ is a bounded simply connected domain where the generating pair (\ref{(F,G)}) is complete and normalized.

\begin{theorem}
Any real-valued continuously differentiable wave function $\Psi=(\psi^{(0)},\psi^{(2)})^\top \in \ker H$ defined for $|z-z_0|<R$ admits a unique expansion of the form
\begin{equation}
\label{psiexpansion}
\Psi=\left(
\begin{array}{c}
\psi^{(0)}\\
\psi^{(2)}
\end{array}
\right)=
\displaystyle \sum_{n=0}^\infty
\left(
\begin{array}{c}
\operatorname{Im} Z^{(n)}(a^{(0)}_n,z_0;z)\\
\operatorname{Re} Z^{(n)}(a^{(2)}_n,z_0;z)
\end{array}
\right),
\end{equation}
$$
\left(
\begin{array}{c}
a^{(0)}_0 \\
a^{(2)}_0
\end{array}
\right)=
\left(
\begin{array}{c}
W^{(0)}(z_0) \\
W^{(2)}(z_0)
\end{array}
\right)\quad \text{ and }\quad \left(
\begin{array}{c}
a^{(0)}_n \\
a^{(2)}_n
\end{array}
\right)=\frac{1}{n!}\frac{\mathrm d_{(F_{n-1},G_{n-1})}}{\mathrm dz}
\left(
\begin{array}{c}
W^{(0)}(z_0) \\
W^{(2)}(z_0)
\end{array}
\right),
$$
where $W^{(0)}=f^{(0)}+\mathrm i \psi^{(0)}$, $f^{(0)}=-\mathrm e^{\chi}\overline{A}[\mathrm i\mathrm e^{-2\chi}\partial_{\overline{z}}(\mathrm e^{\chi}\psi^{(0)})]$ and $W^{(2)}=\psi^{(2)}+\mathrm i f^{(2)}$, $f^{(2)}=\mathrm e^{-\chi}\overline{A}[\mathrm i\mathrm e^{2\chi}\partial_{\overline{z}}(\mathrm e^{-\chi}\psi^{(2)})]$.

Expansion (\ref{psiexpansion}) converges normally for $|z-z_0|<R$.
\end{theorem}
\begin{proof}
Let us first consider the component $\psi^{(0)}$ of $\Psi$. From a continuously differentiable function $\psi^{(0)}$ in the kernel of $H^{(0)}$ we construct a continuously differentiable solution $~{W^{(0)}=f^{(0)}+\mathrm i \psi^{(0)}}$ of the Vekua equation (\ref{mainvekua}), where $f^{(0)}$ is given according to equation (\ref{transfDarbouxinv}) of theorem \ref{PrTransform}. Then from the expansion theorem \ref{Expansion} we obtain  $~{W^{(0)}=\sum_{n=0}^\infty Z^{(n)}(a^{(0)}_n,z_0;z)}$ and first component of the vector in equation (\ref{psiexpansion}) follows directly. Note that in theorem \ref{Expansion} we have $\theta=1$ for the complete and normalized generating pair (\ref{(F,G)}) of the considered domain $\Omega$ (see \cite{BersApprox} for details).

The second component $\psi^{(2)}$ of $\Psi$ is calculated in a similar way.
\end{proof}

\begin{theorem}\cite{KrRecentDevelopments}
In a bounded simply connected domain $\Omega$ such that $\chi\in C^1(\overline \Omega)$ the sets of functions
\begin{equation}
\label{completeH0}
\Big\{\operatorname{Im } Z^{(n)}(1,z_0;z),\ \operatorname{Im } Z^{(n)}(\mathrm i,z_0;z)\Big\}\Big|_{n=0}^{\infty}
\end{equation}
and
\begin{equation}
\label{completeH1}
\Big\{\operatorname{Re } Z^{(n)}(1,z_0;z),\ \operatorname{Re } Z^{(n)}(\mathrm i,z_0;z)\Big\}\Big|_{n=0}^{\infty}
\end{equation}
represent complete systems of solutions for the continuously differentiable functions of the kernels of $H^{(0)}$ and $H^{(1)}$, respectively.
\end{theorem}
\begin{remark} This theorem means that any continuously differentiable element in the kernel of $H^{(0)}$ ($H^{(1)}$) can be represented by a normally convergent sequence of formal polynomials formed by imaginary (real) parts of the functions $Z^{(n)}(1,z_0;z)$ and $Z^{(n)}(\mathrm i,z_0;z)$ in any bounded simply connected domain $\Omega$.
\end{remark}


\section{Separation of variables and transmutation \\ operators}

In pseudoanalytic function theory an important problem is to find the generating sequence (\ref{DefSeq}) such that a generating pair $(F,G)$ of a given Vekua equation be embedded and then use this generating sequence to explicitly construct formal powers (\ref{DefFormalPower}) of the Vekua equation. For Vekua equations in the form of the main Vekua equation (\ref{mainvekua}) some results have been achieved by Bers \cite{Berskniga} and, more recently, by Kravchenko \cite{APFT}. In what follows we present some of these results.

Let us first introduce orthogonal coordinate systems in a plane defined (see \cite{Madelung}) from Cartesian coordinates $x,y$ by means of the relation
$$
u+\mathrm{i}v=\Phi,
$$
where $\Phi=\Phi(x+\mathrm i y)$ is an arbitrary complex analytic function, $u=u(x,y)$ and $v=v(x,y)$. Quite often a transition to more general coordinates is used
$$
\xi=\xi(u) \quad \text{ and }\quad \eta=\eta(v).
$$

\begin{example}
Polar coordinates can be defined as
$$
u+\mathrm i v=\ln(x+\mathrm i y), \quad u=\ln\sqrt{x^2+y^2},\quad v=\arctan\frac{y}{x}.
$$
More frequently polar coordinates are introduced as
$$
r=\mathrm e^u=\sqrt{x^2+y^2},\quad \varphi=v=\arctan\frac{y}{x}.
$$

We can also define parabolic coordinates as
$$
u+\mathrm i v=\sqrt{2(x+\mathrm i y)},\quad u=\sqrt{r+x},\quad v=\sqrt{r-x}.
$$
Usually the following new coordinates are introduced:
$$
\xi=u^2,\quad \eta=v^2.
$$

Other orthogonal coordinate systems (elliptic, bipolar) can be introduced in a similar way (see \cite{APFT, Madelung}).
\end{example}

\begin{theorem}\cite{{APFT},{KrRecentDevelopments}}
Let $F=\mathrm e^{\chi}$ and $G=\mathrm i\mathrm e^{-\chi}$ where $\chi=\chi_1(u)+\chi_2(v)$, $\chi_1,\chi_2$ are arbitrary differentiable real-valued functions, $\Phi=u+\mathrm i v$ is an analytic function of the variable $z=x_1+\mathrm i x_2$ in $\Omega$ such that $\partial_z \Phi$ is bounded and has no zeros in $\Omega$. Then the generating pair $(F,\ G)$ is embedded in $\{(F_m,\ G_m)\}$, a generating sequence in $\Omega$ given by
$$
(F_m,\ G_m)=\left\{
\begin{array}{ll}
\Big((\partial_z \Phi)^m \mathrm e^{\chi_1+\chi_2},\ (\partial_z \Phi)^m \mathrm i\mathrm e^{-(\chi_1+\chi_2)}\Big),& m\text{ even}, \\*[2ex]
\Big((\partial_z \Phi)^m \mathrm e^{-\chi_1+\chi_2},\ (\partial_z \Phi)^m \mathrm i\mathrm e^{\chi_1-\chi_2}\Big),& m\text{ odd},
\end{array}
\right.
$$
for all $m\in\mathbb{Z}$.
\end{theorem}
This theorem opens the way for explicit construction of formal powers corresponding to the main Vekua equation (\ref{mainvekua}) when $\chi$ is separable in terms of orthogonal coordinates $u(x,y)$ and $v(x,y)$ in the plane.

For the Cartesian coordinates Bers obtained elegant explicit formulas for the formal powers \cite{Berskniga}. In this case the separable form of $\chi=\chi_1(x)+\chi_2(y)$ allows us to write down a generating pair $(F,G)=(\mathrm e^{\chi_1+\chi_2},\mathrm i \mathrm e^{-(\chi_1+\chi_2)})$ as well as the generating sequence of period two (except when  $\chi_1(x)\equiv 0$, then the period is one) embedding this generating pair
\begin{equation}
\label{FGF1G1}
\begin{array}{ll}
(F,G)=\big(\mathrm e^{\chi_1+\chi_2},\ \mathrm i\mathrm e^{-(\chi_1+\chi_2)}\big),&  (F_1,G_1)=\big(\mathrm e^{-\chi_1+\chi_2},\ \mathrm i\mathrm e^{\chi_1-\chi_2}\big), \\*[2ex]
(F_2,G_2)=(F,G),&  (F_3,G_3)=(F_1,G_1),\ldots
\end{array}
\end{equation}
For simplicity we assume that $\chi_1(0)=\chi_2(0)=0$. In this case the formal powers are
constructed in an elegant manner as follows. We define the auxiliary functions

\begin{equation}
X^{(0)}(x)\equiv \widetilde{X}^{(0)}(x)\equiv Y^{(0)}(y)\equiv \widetilde{Y}^{(0)}(y)\equiv 1, \label{X1}%
\end{equation}
with
\begin{equation}
X^{(n)}(x)=n%
{\displaystyle\int\limits_{x_{0}}^{x}}
X^{(n-1)}(s)\exp\big[(-1)^{n}2\chi_1(s)\big]\,\mathrm{d}s, \label{X3}%
\end{equation}
\begin{equation}
\widetilde{X}^{(n)}(x)=n%
{\displaystyle\int\limits_{x_{0}}^{x}}
\widetilde{X}^{(n-1)}(s)\exp\big[(-1)^{n+1}2\chi_1(s)\big]\,\mathrm{d}s,
\label{X2}%
\end{equation}
\begin{equation}
Y^{(n)}(y)=n%
{\displaystyle\int\limits_{y_{0}}^{y}}
Y^{(n-1)}(s)\exp\big[(-1)^{n}2\chi_2(s)\big]\,\mathrm{d}s, \label{T3}%
\end{equation}
\begin{equation}
\widetilde{Y}^{(n)}(y)=n%
{\displaystyle\int\limits_{y_{0}}^{y}}
\widetilde{Y}^{(n-1)}(s)\exp\big[(-1)^{n+1}2\chi_2(s)\big]\,\mathrm{d}s,
\label{T2}%
\end{equation}%
where $z_0=(x_0,y_0)$ is an arbitrary fixed point in $\Omega$. Then for $a=a_1+\mathrm{i}a_2$ we have
\begin{equation}
\label{Zn}
Z^{(n)}(a,z_0;z)=\mathrm e^{\chi_1+\chi_2}\,\mathrm{Re}\ _*Z^{(n)}(a,z_0;z)+\mathrm i\mathrm e^{-(\chi_1+\chi_2)}\,\mathrm{Im}\ _*Z^{(n)}(a,z_0;z),
\end{equation}
where
$$
_*Z^{(n)}(a,z_0;z)=\left\{
\begin{tabular}[c]{ll}
$a_1\displaystyle \sum_{k=0}^n\binom{n}{k}X^{(n-k)}\mathrm{i}^k
\widetilde Y^{(k)}+\mathrm{i}a_2\displaystyle
\sum_{k=0}^n\binom{n}{k}\widetilde
X^{(n-k)}\mathrm{i}^k Y^{(k)}$, & odd $n$ \\
\\
$a_1\displaystyle \sum_{k=0}^n\binom{n}{k}\widetilde
X^{(n-k)}\mathrm{i}^k \widetilde Y^{(k)}+\mathrm{i}a_2\displaystyle
\sum_{k=0}^n\binom{n}{k}X^{(n-k)}\mathrm{i}^k Y^{(k)}$, & even $n$.
\end{tabular}
\ \ \ \ \ \ \ \ \ \right.
$$

We introduce the
infinite systems of functions
\begin{equation}
\label{systemsx}
\varphi_{k}(x)=%
\begin{cases}
\mathrm e^{\chi_1(x)}X^{(k)}(x), & k\text{\ odd}\\*[2ex]
\mathrm e^{\chi_1(x)}\widetilde{X}^{(k)}(x), & k\text{\ even}%
\end{cases},\quad
\widetilde \varphi_{k}(x)=%
\begin{cases}
\mathrm e^{-\chi_1(x)} \widetilde X^{(k)}(x), & k\text{\ odd}\\*[2ex]
\mathrm e^{-\chi_1(x)} X^{(k)}(x), & k\text{\ even}%
\end{cases}
\end{equation}
and
\begin{equation}
\label{systemsy}
\psi_{k}(y)=%
\begin{cases}
\mathrm e^{\chi_2(y)}Y^{(k)}(y), & k\text{\ odd}\\*[2ex]
\mathrm e^{\chi_2(y)}\widetilde{Y}^{(k)}(y), & k\text{\ even}%
\end{cases},\quad
\widetilde \psi_{k}(y)=%
\begin{cases}
\mathrm e^{-\chi_2(y)} \widetilde Y^{(k)}(y), & k\text{\ odd}\\*[2ex]
\mathrm e^{-\chi_2(y)} Y^{(k)}(y), & k\text{\ even},%
\end{cases}
\end{equation}
where $k\in \mathbb{Z}_{\geq 0}$.

\begin{example}
\label{ExamplePoly} Consider the case when the superpotential is identically zero, i.e. $\chi\equiv 0$. Then it is easy to see that we have $\varphi_{k}(x)=\widetilde \varphi_{k}(x)=(x-x_0)^{k}$ and $\psi_{k}(y)=\widetilde \psi_{k}(y)=(y-y_0)^{k}$ for $k\in\mathbb{Z}_{\geq 0}$.
\end{example}


Using property 2, following definition \ref{DefFormalPower}, every formal power $Z^{(n)}(a,z_0;z)$ can be expressed in terms of a linear combination of the formal powers $Z^{(n)}(1,z_0;z)$ and $Z^{(n)}(\mathrm i,z_0;z)$  which can be explicitly written in terms of the infinite systems of functions (\ref{systemsx}) and (\ref{systemsy}) as
$$
Z^{(0)}(1,z_0;z)=F=\varphi_0 \psi_0=\mathrm e^{\chi_1+\chi_2},\qquad Z^{(0)}(\mathrm i,z_0;z)=G=\mathrm i\widetilde \varphi_0 \widetilde \psi_0=\mathrm i\mathrm e^{-(\chi_1+\chi_2)},
$$
and, for $n>0$,
\small
\begin{equation}
\label{Z1}
Z^{(n)}(1,z_0;z)=\left\{
\begin{array}{ll}
\displaystyle \sum_{k=0}^{\frac{1}{2}(n-1)}(-1)^k\binom{n}{2k}\varphi_{n-2k}\psi_{2k}+\mathrm i\displaystyle \sum_{k=0}^{\frac{1}{2}(n-1)}(-1)^k\binom{n}{2k+1}\widetilde \varphi_{n-2k-1}\widetilde \psi_{2k+1},& n\mbox{ odd}\\*[2ex]
\displaystyle \sum_{k=0}^{n/2}(-1)^k\binom{n}{2k}\varphi_{n-2k} \psi_{2k}+\mathrm i\displaystyle \sum_{k=0}^{\frac{1}{2}(n-2)}(-1)^k\binom{n}{2k+1}\widetilde \varphi_{n-2k-1}\widetilde \psi_{2k+1},& n\mbox{ even}
\end{array}
\right.
\end{equation}
\begin{equation}
\label{Zi}
Z^{(n)}(\mathrm i,z_0;z)=\left\{
\begin{array}{ll}
\displaystyle \sum_{k=0}^{\frac{1}{2}(n-1)}(-1)^{k+1}\binom{n}{2k+1}\varphi_{n-2k-1} \psi_{2k+1}+\mathrm i\displaystyle \sum_{k=0}^{\frac{1}{2}(n-1)}(-1)^k\binom{n}{2k}\widetilde \varphi_{n-2k}\widetilde \psi_{2k},& n\mbox{ odd}\\*[2ex]
\displaystyle \sum_{k=0}^{\frac{1}{2}(n-2)}(-1)^{k+1}\binom{n}{2k+1}\varphi_{n-2k-1} \psi_{2k+1}+\mathrm i \displaystyle \sum_{k=0}^{n/2}(-1)^k\binom{n}{2k}\widetilde \varphi_{n-2k}\widetilde \psi_{2k},& n\mbox{ even}.
\end{array}
\right.
\end{equation}
\normalsize

The $(F,G)$-derivative of formal powers $Z^{(n)}(a,z_0;z)$ can be obtained using theorem \ref{FGderivativetheo}, i.e. $\overset{\circ}{Z}^{(n)}\!\!(a,z_0;z)=\partial_z Z^{(n)}(a,z_0;z)-(\partial_z \chi)\overline {Z^{(n)}}(a,z_0;z)$. However,  it is more efficient to remark that
$$
\overset{\circ}{Z}^{(n)}\!\!(a,z_0;z)=\frac{\mathrm d_{(F,G)}}{\mathrm dz}Z^{(n)}(a,z_0;z)=n Z^{(n-1)}_1(a,z_0;z)
$$
is an $(F_1,G_1)$-pseudoanalytic functions satisfying the Vekua equation (\ref{mainvekua}) for $\chi=-\chi_1+\chi_2$ from the generating pairs (\ref{FGF1G1}). Then formal powers $Z_1^{(n)}(a,z_0;z)$ can be constructed by considering transformations
$$
\pm\chi_1\longrightarrow \mp \chi_1 \qquad \text{ and } \qquad \pm\chi_2\longrightarrow \pm \chi_2\quad (\chi_2\text{ invariant})
$$
in the formal powers $Z^{(n)}(a,z_0;z)$. These transformations imply that
$$
\varphi_k \longrightarrow \widetilde \varphi_k, \quad \widetilde \varphi_k \longrightarrow \varphi_k, \quad \psi_k \longrightarrow \psi_k,\quad \widetilde\psi_k \longrightarrow \widetilde\psi_k\quad (\psi_k,\ \widetilde \psi_k \text{ invariants})
$$
such that
$$
Z_1^{(0)}\!\!(1,z_0;z)=F_1=\widetilde \varphi_0 \psi_0=\mathrm e^{-\chi_1+\chi_2},\quad Z_1^{(0)}\!\!(\mathrm i,z_0;z)=G_1=\mathrm i\varphi_0\widetilde \psi_0=\mathrm i\mathrm e^{\chi_1-\chi_2},
$$
and for $n>0$ we have
\small
\begin{equation}
\label{Z11}
Z_1^{(n)}(1,z_0;z)=\left\{
\begin{array}{ll}
\displaystyle \sum_{k=0}^{\frac{1}{2}(n-1)}(-1)^k\binom{n}{2k}\widetilde \varphi_{n-2k} \psi_{2k}+\mathrm i\displaystyle \sum_{k=0}^{\frac{1}{2}(n-1)}(-1)^k\binom{n}{2k+1}\varphi_{n-2k-1}\widetilde \psi_{2k+1},& n\mbox{ odd}\\*[2ex]
\displaystyle \sum_{k=0}^{n/2}(-1)^k\binom{n}{2k}\widetilde \varphi_{n-2k} \psi_{2k}+\mathrm i\displaystyle \sum_{k=0}^{\frac{1}{2}(n-2)}(-1)^k\binom{n}{2k+1}\varphi_{n-2k-1}\widetilde \psi_{2k+1},& n\mbox{ even}
\end{array}
\right.
\end{equation}
\begin{equation}
\label{Z1i}
Z_1^{(n)}(\mathrm i,z_0;z)=\left\{
\begin{array}{ll}
\displaystyle \sum_{k=0}^{\frac{1}{2}(n-1)}(-1)^{k+1}\binom{n}{2k+1}\widetilde \varphi_{n-2k-1} \psi_{2k+1}+\mathrm i\displaystyle \sum_{k=0}^{\frac{1}{2}(n-1)}(-1)^k\binom{n}{2k}\varphi_{n-2k}\widetilde \psi_{2k},& n\mbox{ odd}\\*[2ex]
\displaystyle \sum_{k=0}^{\frac{1}{2}(n-2)}(-1)^{k+1}\binom{n}{2k+1}\widetilde \varphi_{n-2k-1} \psi_{2k+1}+\mathrm i \displaystyle \sum_{k=0}^{n/2}(-1)^k\binom{n}{2k}\varphi_{n-2k}\widetilde \psi_{2k},& n\mbox{ even},
\end{array}
\right.
\end{equation}
\normalsize
where $Z_1^{(n)}(a,z_0;z)=\frac{1}{n+1}\overset{\circ}{Z}^{(n+1)}(a,z_0;z)$.

In \cite{CKT} it was shown that if the functions $\chi_j\in C^2(-a_j,a_j)\cap C^1[-a_j,a_j]$ are such that $\chi_j(0)=0$ and $\chi_j$ is bounded on $[-a_j,a_j]$, where $j\in\{1,2\}$, there exist the transmutation operators $T_1,T_2$ defined as follows
$$
T_{j}[f(x_j)]=f(x_j)+\int_{-x_j}^{x_j} \textbf{K}_j(x_j,t;h_j)\cdot f(t)\mathrm dt,\qquad j\in\{1,2\},
$$
where $h_j=\partial_j\mathrm e^{\chi_j}\big|_{x_j=0}$, the kernel $\textbf{K}_j(x_j,s;h_j)$ is given by
$$
\textbf{K}_j(x_j,t;h_j)=\frac{h_j}{2}+K_j(x_j,t)+\frac{h_j}{2}\int_{t}^{x_j} [K_j(x_j,s)-K_j(x_j,-s)]\mathrm ds
$$
and the function $K_j(x_j,t)$ is the unique solution of the Goursat problem \cite{KKTT, KT, Marchenko}
$$
\Big(\partial_j^{2}-q_j(x_j)\Big)  K_j(x_j,t)=\frac
{\partial^{2}}{\partial t^{2}}K_j(x_j,t), \quad q_j=\partial_j^2 \chi_j+(\partial_j \chi_j)^2,\label{Goursat1}%
$$
$$
K_j(x_j,x_j)=\frac{1}{2}\int_{0}^{x_j}q_j(s)\mathrm ds,\qquad K_j(x_j,-x_j)=0. \label{Goursat1}%
$$

Moreover, $T_1$ and $T_2$ satisfy the relations
\begin{equation}
\label{T1xT2y}
T_1[x^k]=\varphi_k \quad \text{ and }\quad T_2[y^k]=\psi_k,\quad k\in \mathbb{Z}_{\geq 0}.
\end{equation}

Another pair of transmutations $\widetilde T_1$ and $\widetilde T_2$ is constructed, one of the representations of which can be given by the equalities
$$
\widetilde T_j[f(x_j)]=\mathrm e^{-\chi_j(x_j)}\left(\int_0^{x_j} \mathrm e^{\chi_j(s)}T_j[\partial_j f(s)]\mathrm ds+f(0)\right),\qquad j\in\{1,2\}.
$$
These operators satisfy the equalities
\begin{equation}
\label{T1tildexT2y}
\widetilde T_1[x^k]=\widetilde \varphi_k \quad \text{ and }\quad \widetilde T_2[y^k]=\widetilde \psi_k,\quad k\in \mathbb{Z}_{\geq 0}.
\end{equation}

The operators $\widetilde T_1$ and $\widetilde T_2$ admit the representations as Volterra integral operators
$$
\widetilde T_j[f(x_j)]=f(x_j)+\int_{-x_j}^{x_j}\widetilde{\textbf{K}}_j(x_j,t;-h_j)\cdot f(t)\mathrm dt,
$$
where the kernel $\widetilde{\textbf{K}}_j(x_j,t;-h_j)$ has the form
$$
\widetilde{\textbf{K}}_j(x_j,t;-h_j)=-\mathrm e^{-\chi_j(x_j)}\left(\int_{-t}^{x_j}\frac{\partial}{\partial t}\textbf{K}_j(s,t;h_j)\cdot \mathrm e^{\chi_j(s)}\mathrm d s+\frac{h_j}{2}\mathrm e^{\chi_j(-t)}\right).
$$

\begin{corollary}
\cite{KT}
\label{partialeT}
The following four operator equalities hold on $C^1[-a_j,a_j]$-functions of the respective variable
\begin{equation}
\label{relationsTpartial}
\partial_j \mathrm e^{\chi_j}\widetilde T_j=\mathrm e^{\chi_j}T_j\partial_j\quad \text{ and } \quad \partial_j \mathrm e^{-\chi_j}T_j=\mathrm e^{-\chi_j}\widetilde T_j\partial_j,
\end{equation}
for $j\in\{1,2\}$.
\end{corollary}
The integral counterpart of this corollary can be considered. Indeed, by applying the operator equality (\ref{relationsTpartial}) on the antiderivative $F_j=\int f(x_1,x_2) \,\mathrm dx_j$ for $j\in \{1,2\}$ and then integrating on $x_j$, we obtain
$$
\mathrm e^{\chi_j}\widetilde T_j F_j=\int \mathrm e^{\chi_j}T_j\partial_j F_j \,\mathrm dx_j\quad \text{ and }\quad
\mathrm e^{-\chi_j} T_j F_j=\int \mathrm e^{-\chi_j}\widetilde T_j\partial_j F_j \,\mathrm dx_j,
$$
i.e.
\begin{equation}
\label{intTd1}
\mathrm e^{\chi_j}\widetilde T_j\int f\,\mathrm dx_j=\int \mathrm e^{\chi_j}T_jf \,\mathrm dx_j
\end{equation}
and
\begin{equation}
\label{intTd2}
\mathrm e^{-\chi_j} T_j\int f\,\mathrm dx_j=\int \mathrm e^{-\chi_j}\widetilde T_jf \,\mathrm dx_j,
\end{equation}
for any continuous real-valued function $f$ on $[-a_j,a_j]$ for the variable $x_j$.

We introduce the following operators
$$
\mathbf{T}_0=T_1T_2 P^++\mathrm i \widetilde T_1\widetilde T_2 P^-
$$
and
$$
\mathbf{T}_1=\widetilde T_1 T_2 P^++\mathrm i T_1 \widetilde T_2 P^-.
$$
From now on let $\Omega\subset \overline R=[-a_1,a_1]\times [-a_2,a_2]$ be a simply connected domain such that together with any point $(x,y)$ belonging to $\Omega$ the rectangle with the vertices $(x,y)$, $(-x,y)$, $(x,-y)$ and $(-x,-y)$ also belongs to $\Omega$. In such a domain application of operators $\mathbf{T}_0$ and $\mathbf{T}_1$ is meaningful.

Let us now consider two results obtained in a recent paper \cite{CKM}.
\begin{theorem}
\cite{CKM}
For any $z_0,z\in \Omega$ and $a=a_1+\mathrm i a_2\in \mathbb{Z}$ the following equalities hold
$$
\mathbf{T}_0[az^n]=Z^{(n)}(a,z_0;z)\quad \text{ and }\quad \mathbf{T}_1[az^n]=Z_1^{(n)}(a,z_0;z),
$$
where
$$
Z^{(n)}(a,z_0;z)=a_1Z^{(n)}(1,z_0;z)+a_2Z^{(n)}(\mathrm i,z_0;z)$$
$$
Z_1^{(n)}(a,z_0;z)=a_1Z_1^{(n)}(1,z_0;z)+a_2Z_1^{(n)}(\mathrm i,z_0;z)
$$
for $Z^{(n)}(1,z_0;z)$, $Z^{(n)}(\mathrm i,z_0;z)$, $Z_1^{(n)}(1,z_0;z)$ and $Z_1^{(n)}(\mathrm i,z_0;z)$ given, respectively, by equations (\ref{Z1}), (\ref{Zi}), (\ref{Z11}) and (\ref{Z1i}).
\end{theorem}

\begin{theorem}
\cite{CKM}
\label{relWw}
For any complex-valued continuously differentiable function $w$ defined in $\Omega$ the following operator equalities hold:
\begin{equation}
\label{relVekuaCR1}
V\mathbf{T}_0=\mathbf{T}_1\partial_{\overline z},\quad  V_1\mathbf{T}_1=\mathbf{T}_0\partial_{\overline z},
\end{equation}
\begin{equation}
\label{relVekuaCR2}
\frac{\mathrm d_{(F,G)}}{\mathrm dz}\mathbf T_0=\mathbf{T}_1\partial_{z},\quad  \frac{\mathrm d_{(F_1,G_1)}}{\mathrm dz}\mathbf T_1=\mathbf{T}_0\partial_{z}.
\end{equation}
\end{theorem}

From equalities (\ref{relVekuaCR1}), we observe that operator $\mathbf T_0$ maps complex analytic function into $(\mathrm e^{\chi_1+\chi_2},\mathrm i\mathrm e^{-(\chi_1+\chi_2)})$-pseudoanalytic function, i.e. into solutions of the Vekua equation (\ref{mainvekua}). In the same way, operator $\mathbf T_1$ maps complex analytic function into $(\mathrm e^{-\chi_1+\chi_2},\mathrm i\mathrm e^{\chi_1-\chi_2})$-pseudoanalytic function, i.e. into solution of the Vekua equation $V_1W=0$.

The integral counterpart is given in the following.

\begin{theorem}
\label{relWwint}
For any continuous complex-valued function $w$ defined in $\Omega$ the following equalities hold:
\begin{equation}
\label{relVekuaCR1int}
\int_{\Gamma}\mathbf T_0[w]\,\mathrm d_{(F_1,G_1)}\zeta=\mathbf T_1\big[\int_{\Gamma}w\,\mathrm d\zeta\big]
\end{equation}
\begin{equation}
\label{relVekuaCR2int}
\int_{\Gamma}\mathbf T_1[w]\,\mathrm d_{(F,G)}\zeta=\mathbf T_0\big[\int_{\Gamma}w\,\mathrm d\zeta\big],
\end{equation}
where $\Gamma$ is a rectifiable curve in $\Omega$.
\end{theorem}
\begin{proof}
Let us consider equation (\ref{relVekuaCR1int}) with $w=w_1+\mathrm i w_2$ and $\zeta=\xi+\mathrm i\eta$. We have $(F_1,G_1)=(\mathrm e^{-\chi_1+\chi_2},\mathrm i\mathrm e^{\chi_1-\chi_2})$ and $(F_1^*,G_1^*)=(-\mathrm i\mathrm e^{-\chi_1+\chi_2},\mathrm e^{\chi_1-\chi_2})$ such that from (\ref{FGintegration})
$$
\begin{array}{rcl}
\displaystyle \int_{\Gamma}\mathbf T_0[w]\,\mathrm d_{(F_1,G_1)}\zeta&=& \mathrm e^{-\chi_1(x)+\chi_2(y)}\operatorname{Re}\displaystyle \int_{\Gamma}\mathrm e^{\chi_1-\chi_2}\mathbf T_0[w]\mathrm d\zeta\\*[2ex]&+&\mathrm i \mathrm e^{\chi_1(x)-\chi_2(y)}\operatorname{Im}\displaystyle \int_{\Gamma}\mathrm e^{-\chi_1+\chi_2}\mathbf T_0[w]\mathrm d\zeta,
\end{array}
$$
where $\Gamma$ is a rectifiable curve leading from $z_0$ to $z=x+\mathrm i y$ in $\Omega$. We obtain
$$
\begin{array}{rcl}
\displaystyle \int_{\Gamma}\mathbf T_0[w]\,\mathrm d_{(F_1,G_1)}\zeta &=& \mathrm e^{-\chi_1(x)+\chi_2(y)}\displaystyle \int_{\Gamma}\mathrm e^{\chi_1-\chi_2}\big(T_1T_2w_1\mathrm d\xi-\widetilde T_1\widetilde T_2 w_2\mathrm d\eta\big)\\ &+& \mathrm i \mathrm e^{\chi_1(x)-\chi_2(y)}\displaystyle \int_{\Gamma}\mathrm e^{-\chi_1+\chi_2}\big(T_1T_2 w_1\mathrm d\eta+\widetilde T_1\widetilde T_2 w_2\mathrm d\xi\big).
\end{array}
$$
Now since operators $T_1,T_2$ commute, as well as $\widetilde T_1,\widetilde T_2$, we find
$$
\begin{array}{rcl}
\displaystyle \int_{\Gamma}\mathbf T_0[w]\,\mathrm d_{(F_1,G_1)}\zeta &=& \mathrm e^{-\chi_1(x)+\chi_2(y)}\displaystyle \int_{\Gamma}\Big(\mathrm e^{-\chi_2}T_2\mathrm e^{\chi_1}T_1w_1\mathrm d\xi-\mathrm e^{\chi_1}\widetilde T_1\mathrm e^{-\chi_2}\widetilde T_2 w_2\mathrm d\eta\Big)\\ && + \mathrm i \mathrm e^{\chi_1(x)-\chi_2(y)}\displaystyle \int_{\Gamma}\Big(\mathrm e^{-\chi_1}T_1\mathrm e^{\chi_2}T_2w_1\mathrm d\eta+\mathrm e^{\chi_2}\widetilde T_2\mathrm e^{-\chi_1}\widetilde T_1 w_2\mathrm d\xi\Big) \\*[2ex]
&=& \mathrm e^{-\chi_1(x)+\chi_2(y)}\Big(\mathrm e^{-\chi_2}T_2\displaystyle \int_{\Gamma} \mathrm e^{\chi_1}T_1w_1\mathrm d\xi-\mathrm e^{\chi_1}\widetilde T_1 \displaystyle \int_{\Gamma} \mathrm e^{-\chi_2}\widetilde T_2 w_2\mathrm d\eta\Big)\\ && + \mathrm i \mathrm e^{\chi_1(x)-\chi_2(y)}\Big(\mathrm e^{-\chi_1}T_1\displaystyle \int_{\Gamma} \mathrm e^{\chi_2}T_2w_1\mathrm d\eta+\mathrm e^{\chi_2}\widetilde T_2\displaystyle \int_{\Gamma} \mathrm e^{-\chi_1}\widetilde T_1 w_2\mathrm d\xi\Big) \\*[2ex]
&=& \mathrm e^{-\chi_1(x)+\chi_2(y)}\Big(\mathrm e^{-\chi_2}T_2\mathrm e^{\chi_1}\widetilde T_1\displaystyle \int_{\Gamma} w_1\mathrm d\xi-\mathrm e^{\chi_1}\widetilde T_1 \mathrm e^{-\chi_2} T_2 \displaystyle \int_{\Gamma}  w_2\mathrm d\eta\Big)\\ && + \mathrm i \mathrm e^{\chi_1(x)-\chi_2(y)}\Big(\mathrm e^{-\chi_1}T_1\mathrm e^{\chi_2}\widetilde T_2 \displaystyle \int_{\Gamma} w_1\mathrm d\eta+\mathrm e^{\chi_2}\widetilde T_2\mathrm e^{-\chi_1} T_1 \displaystyle \int_{\Gamma} w_2\mathrm d\xi\Big),
\end{array}
$$
where relations (\ref{intTd1}), (\ref{intTd2}) have been used in the last equality. Hence we obtain
$$
\displaystyle \int_{\Gamma}\mathbf T_0[w]\,\mathrm d_{(F_1,G_1)}\zeta=\widetilde T_1T_2 \operatorname{Re}\int_{\Gamma}w\mathrm d\zeta+
\mathrm i T_1\widetilde T_2 \operatorname{Im}\int_{\Gamma}w\mathrm d\zeta=\mathbf T_1\big[\int_{\Gamma}w\mathrm d\zeta\big].
$$

The other relation (\ref{relVekuaCR2int}) is shown in a similar way.
\end{proof}

Theorems \ref{relWw}, \ref{relWwint} and \ref{solutionVekuaSusy} can be summarized up in the following commutative diagrams. Let $w$ be  a continuously differentiable complex-valued function in $\Omega$ and $\chi$ separable, i.e. $\chi=\chi_1(x)+\chi_2(y)$:
$$
\begin{array}{rcrcl}
\mathbf T_1[w] & \overset{\mathbf T_1}{\longleftarrow} & w & \overset{\mathbf T_0}{\longrightarrow} & \mathbf T_0[w] \\*[2ex]
V_1 \downarrow && \partial_{\overline z} \downarrow && \downarrow V\\*[2ex]
\mathbf T_0[\partial_{\overline z}w] & \overset{\mathbf T_0}{\longleftarrow} & \partial_{\overline z}w & \overset{\mathbf T_1}{\longrightarrow} & \mathbf T_1[\partial_{\overline z}w].
\end{array}
$$

Moreover, when $w$ is analytic in $\Omega$ we obtain the following results:
 \small
$$
\begin{array}{rcrcrclcl}
 P\mathbf T_0[\int w \,\mathrm d\zeta]\in \ker H & \overset{P}{\longleftarrow} & \mathbf T_0[\int w \,\mathrm d\zeta] &  \overset{\mathbf T_0}{\longleftarrow} & \int w \,\mathrm d\zeta & \overset{\mathbf T_1}{\longrightarrow} & \mathbf T_1[\int w \,\mathrm d\zeta] & \overset{P}{\longrightarrow} & P\mathbf T_1[\int w \,\mathrm d\zeta]\in \ker H^{(1)}\\*[2ex]
 \downarrow \frac{1}{2}D && (F,G)\text{-}\int \uparrow  & &\int \uparrow && \uparrow (F_1,G_1)\text{-}\int && \downarrow \frac{1}{2}D_1\\*[2ex]
 P\mathbf T_1[w] \in \ker H^{(1)}& \overset{P}{\longleftarrow} & \mathbf T_1[w]& \overset{\mathbf T_1}{\longleftarrow}& w & \overset{\mathbf T_0}{\longrightarrow} & \mathbf T_0[w] & \overset{P}{\longrightarrow} & P\mathbf T_0[w]\in \ker H\\*[2ex]
\downarrow \frac{1}{2}D_1 && \frac{\mathrm d_{(F_1,G_1)}}{\mathrm dz} \downarrow && \partial_z \downarrow && \downarrow \frac{\mathrm d_{(F,G)}}{\mathrm dz} && \downarrow \frac{1}{2}D\\*[2ex]
  P\mathbf T_0[\partial_z w] \in \ker H& \overset{P}{\longleftarrow} &\mathbf T_0[\partial_z w] &\overset{\mathbf T_0}{\longleftarrow}& \partial_z w & \overset{\mathbf T_1}{\longrightarrow} & \mathbf T_1[\partial_z w]& \overset{P}{\longrightarrow} & P\mathbf T_1[\partial_z w]\in \ker H^{(1)},
\end{array}
$$
\normalsize
where $D$, $D_1$ are the Darboux and pseudo-Darboux transformations defined by (\ref{2DDarbouxTransf}) and (\ref{D1}), respectively.

\begin{corollary}
For any complex-valued  function in $C^2(\Omega)$ the following operator equalities hold:
$$
HP\mathbf T_0=-P\mathbf T_0\Delta \qquad \text{ and }\qquad H^{(1)}P\mathbf T_1=-P\mathbf T_1\Delta.
$$
\end{corollary}
\begin{proof}
Using (\ref{HP=4PVV}), (\ref{relVekuaCR1}) and (\ref{relVekuaCR2}) we find
$$
HP\mathbf T_0=-4PV_1\overline V\mathbf T_0=-4PV_1\mathbf T_1\partial_z=-4P\mathbf T_0\partial_{\overline z}\partial_z=-P\mathbf T_0\Delta.
$$
The other relation is obtained in a similar way.
\end{proof}

\bigskip

\section{Conclusion}

Relations between Vekua and Bers derivatives operators were obtained in terms of two-dimensional Darboux and pseudo-Darboux transformations. These relations enable us to factorize hamiltonians $H^{(0)}$ and $H^{(1)}$ of the superhamiltonian $\widehat H$ considered in this work in terms of one Vekua operator and one Bers derivative operator. An infinite system of solutions of the ground state for the two-dimensional SUSY QM system was given in terms of the formal powers. Under certain additional assumptions, completeness results for the system of formal powers, analogous to the expansion theorem and Runge's approximation theorem from classical analysis, were established. For the specific case where the superpotential is separable, formal powers were explicitly constructed. That specific case allowed us to use transmutation operators to transform complex analytic powers to the constructed formal powers. Moreover, we have shown how transmutation operators are related to Vekua, Bers derivative, Bers integral operators and hamiltonian components of the superhamiltonian.

Our approach can also be applied to the two-dimensional SUSY QM system considered in this work with the superpotential $\chi$ being a complex
function, though in this case complex numbers become insufficient, and one should consider
the bicomplex pseudoanalytic function theory \cite{CK}.

\subsubsection*{Acknowledgement}

A.B. acknowledges two scholarships: one from NSERC, where a part of this work was done, and one from the Institut des Sciences Math\'ematiques (ISM). The research of S.T. is partly supported by grant from NSERC of Canada.

\end{document}